\documentclass[aps,prd,preprint,superscriptaddress,longbibliography,nofootinbib,floatfix]{revtex4-2}
\usepackage{CJKutf8}
\usepackage{booktabs}
\usepackage{graphicx}
\usepackage{amsmath,amssymb,amsthm}
\usepackage{mathrsfs}
\usepackage[hidelinks]{hyperref}
\usepackage{orcidlink}
\numberwithin{equation}{section}

\newtheorem{theorem}{Theorem}[section]
\newtheorem{proposition}[theorem]{Proposition}
\newtheorem{lemma}[theorem]{Lemma}
\newtheorem{corollary}[theorem]{Corollary}

\theoremstyle{definition}

\theoremstyle{remark}

\begin{document}

\title{Convergence Without Finite Exactness in the Eigenstate Quantum Bootstrap}

\author{Ye Zhou (\begin{CJK*}{UTF8}{gbsn}周烨\end{CJK*})\orcidlink{0009-0004-7050-7736}}
\email{ye.zhou.horizon@gmail.com}
\affiliation{Independent Researcher, Kunshan, Jiangsu, China}
\author{Daniel Keren}
\email{dkeren@ds.haifa.ac.il}
\affiliation{Department of Computer Science, University of Haifa, Haifa, Israel}

\begin{abstract}
The eigenstate quantum-mechanical bootstrap can converge to the physical moment set while remaining nonexact at every finite order. We prove this for a one-dimensional quasi-exactly solvable sextic oscillator at explicitly known eigenvalues, including excited levels. At every finite order, the lower and upper bounds on $u^2$ remain strictly below and above its physical value, even after adjoining arbitrary finite polynomial Gram tests and all eigenstate-annihilator relations visible in the retained moment space. Hence neither side of the centered observable admits a finite sum-of-squares certificate modulo the eigenstate annihilator. The gap persists under any finite set of additional localizing constraints whose weights are strictly positive on polynomial vectors.  We then prove convergence for polynomial Schr\"odinger operators satisfying explicit sum-of-squares confinement and derivative bounds. The eigenstate relations give moment estimates uniform in the truncation order; these yield a regular Schr\"odinger representation supported on the prescribed eigenspace. Finally, a four-mode extension combines this fixed-observable obstruction with feasible position moments having no positive representing measure at every finite order. At level three, this construction yields nonphysical feasible points with strictly positive definite moment matrices.
\end{abstract}

\maketitle

\section{Introduction}\label{Sec:Intro}

Moment methods constrain quantum spectra through positivity and operator identities \cite{Handy1985,Handy1987}. The quantum-mechanical bootstrap implements these conditions as semidefinite problems for finitely many expectation values \cite{Han2020,Berenstein2021,Bhattacharya2021,Berenstein2022}. Energy-conservation formulations give bounds at prescribed energy \cite{Berenstein2023,Morita2023}; related hierarchies bound stationary observables in open bosonic systems \cite{RobichonTilloy2026}, while spectral crossing equations constrain gaps and matrix elements \cite{NancarrowXin2023}. These methods make it possible to obtain useful bounds before reconstructing a wavefunction.

In the Kronig--Penney model, Blacker, Bhattacharyya and Banerjee found that position-based Toeplitz constraints reproduce band gaps but leave the minimum energy poorly constrained; including higher momentum powers improves the latter bounds \cite{Blacker2022}. In earlier work, one of the authors \cite{Keren2026} exhibited a certified nonphysical feasible point for the stationary bootstrap and investigated the stronger eigenstate constraints. A companion paper gives spurious level-3 eigenstate-bootstrap moment matrices that are strictly positive definite at the first excited and ground-state energies of a coupled three-dimensional double well \cite{KerenZhou2026Interior}. These examples motivate a distinction between the existence of spurious finite moment data, the exactness of a bound on a fixed observable, and convergence as the truncation order grows.

Finite and asymptotic convergence are already distinct in commutative moment--sum-of-squares (SOS) optimization \cite{Lasserre2001,Parrilo2003,Nie2014}. For canonical quantum variables there is a further issue: position and momentum are unbounded. Archimedean noncommutative optimization \cite{Pironio2010} and positivity theorems for free polynomials \cite{Helton2002} do not by themselves identify a regular representation of the canonical commutation relations. Bosonic reconstruction requires control of moment growth \cite{Navascues2013}. We derive that control from the eigenstate equations under explicit polynomial hypotheses.

Our finite-order result already occurs in one dimension. For a sextic quasi-exactly solvable oscillator, the physical expectation of $u^2$ is known at three simple eigenvalues. At each of these energies, every finite relaxation built from polynomial Gram tests and visible eigenstate-annihilator relations admits values on both sides of the physical value. Neither sign of the centered observable therefore has a finite polynomial-square certificate modulo that annihilator. The same separation persists for finite localizing tests whose weights are strictly positive on polynomial vectors. The proof uses a one-parameter deformation of the moments of a quartic weight: all polynomial eigenstate identities remain valid, and positivity at any fixed finite degree survives near the physical parameter.

We also prove that the hierarchy converges to the physical energy shell for a class of confining polynomial Schr\"odinger operators. Uniform moment estimates give compactness and analytic vectors in the limiting Gelfand--Naimark--Segal (GNS) representation. Integration of the canonical relations and the Stone--von Neumann theorem then produce a density operator on the prescribed eigenspace. The sextic oscillator satisfies the hypotheses, so its finite bounds are always strict and converge to the physical value. A four-mode extension supplies a complementary example in which the position moments themselves have no positive representing measure at any finite order.

Section~\ref{Sec:OneD} proves the one-dimensional obstruction and its persistence under finite strengthenings. Section~\ref{Sec:Convergence} establishes reconstruction and convergence. Section~\ref{Sec:Four} combines the two finite-order constructions on an explicitly determined excited energy shell. The remaining questions concern weaker convergence hypotheses and quantitative error estimates.

\section{Finite nonexactness in one dimension}\label{Sec:OneD}

\subsection{The eigenstate hierarchy}\label{Sec:Hierarchy}

Let $\mathcal A_d$ be the polynomial Weyl $*$-algebra generated by $d$ canonical pairs, with
\[
[q_j,p_k]=\mathrm i\delta_{jk},\qquad [q_j,q_k]=[p_j,p_k]=0,\qquad q_j^\dagger=q_j,\quad p_j^\dagger=p_j.
\]
We use Poincar\'e--Birkhoff--Witt (PBW) order, with all positions to the left of all momenta, and write $\mathcal A_{d,\le n}$ for the polynomials of total degree at most $n$. Consider
\begin{equation}\label{eq:general-H}
H_V=\frac12\sum_{j=1}^d p_j^2+V(q),\qquad \deg V=2r,\quad r\ge1.
\end{equation}
For $K\ge r$, let $\mathcal F_K^V(E)$ consist of complex linear functionals $L:\mathcal A_{d,\le2K}\to\mathbb C$ satisfying
\begin{align}
L(1)&=1,\qquad L(A^\dagger)=\overline{L(A)},\nonumber\\
M_K(L)&:=\bigl(L(v_a^\dagger v_b)\bigr)_{a,b}\succeq0,\nonumber\\
L([H_V,O])&=0
\quad\text{for every }O\in\mathcal A_d\text{ with }[H_V,O]\in\mathcal A_{d,\le2K},\label{eq:finite-hierarchy}\\
L(O(H_V-E))&=0
\quad\text{for every }O\in\mathcal A_{d,\le2K-2r}.\nonumber
\end{align}
Here $\{v_a\}$ is the PBW monomial basis through degree $K$. The condition $M_K(L)\succeq0$ means $L(A^\dagger A)\ge0$ for $\deg A\le K$. The commutator condition includes every polynomial $O$ whose reduced commutator is retained, including combinations with cancellation of higher-degree terms. Hermiticity also gives the corresponding left energy relations.

A physical functional is $L(A)=\operatorname{Tr}(\rho A)$, where $\rho$ is a density operator supported on $\ker(H_V-E)$ and has the required moments. Its restriction is feasible at every level. We ask when the level-$K$ feasible set \emph{strictly} contains these physical moments.

\subsection{The sextic oscillator}

Set $q=u$, $p=p_u$ and
\begin{equation}\label{eq:oneD-H}
h_u=\frac12(p_u^2+u^6-11u^2).
\end{equation}
Its smooth potential is bounded below, so the operator is essentially self-adjoint on $C_c^\infty(\mathbb R)$ \cite{Simon1973}. Since the potential tends to $+\infty$, its resolvent is compact. Define
\begin{equation}\label{eq:phi-eta}
\phi(u)=C e^{-u^4/4},\qquad
\mu(f)=\int_{\mathbb R}f(u)|\phi(u)|^2\,du,\qquad
\eta=\mu(u^2)=\sqrt2\,\frac{\Gamma(3/4)}{\Gamma(1/4)},
\end{equation}
with $\|\phi\|_2=1$. Integration by parts gives
\begin{equation}\label{eq:moment-recurrence}
\mu(u^{n+4})=\frac{n+1}{2}\mu(u^n),\qquad n\ge0,
\end{equation}
and all odd moments vanish. Conjugating by $\phi$ yields
\begin{equation}\label{eq:gauge-rotation}
h_u(P\phi)=(\mathscr DP)\phi,\qquad
\mathscr D=-\frac12\partial_u^2+u^3\partial_u-4u^2.
\end{equation}
The invariant space $\operatorname{span}\{1,u^2,u^4\}$ gives the standard quasi-exactly solvable sector \cite{TurbinerUshveridze1987,Turbiner1988,Turbiner2016}:
\begin{align}
P_-(u)&=u^4+2u^2+\tfrac12,&\mathscr DP_-&=-4P_-,\nonumber\\
P_0(u)&=u^4-\tfrac32,&\mathscr DP_0&=0,\label{eq:qes-states}\\
P_+(u)&=u^4-2u^2+\tfrac12,&\mathscr DP_+&=4P_+.\nonumber
\end{align}
For $e\in\{-4,0,4\}$, write $P_e=P_-,P_0,P_+$ respectively, and
\[
\psi_e=\frac{P_e\phi}{\sqrt{\mu(P_e^2)}}.
\]
The polynomials have zero, two and four real zeros. Whole-line Sturm oscillation therefore identifies their energies as $\lambda_0=-4$, $\lambda_2=0$ and $\lambda_4=4$, where $\lambda_0<\lambda_1<\cdots$ is the spectrum of $h_u$. The nodal-count argument applies to this confining potential: restrictions of an eigenfunction to its nodal intervals lie in the quadratic-form domain, so Sturm comparison and the min--max principle give its spectral index; compare the proof in Ref.~\cite[Sec.~9.7]{Teschl2014}. In particular, the latter two states are excited. The spectrum is simple, so the only density operator at energy $e$ is $|\psi_e\rangle\langle\psi_e|$.

The recurrence gives
\[
\mu(P_-^2)=4+8\eta,\qquad \mu(P_0^2)=2,\qquad
\mu(P_+^2)=4-8\eta,
\]
and hence $0<\eta<1/2$. The physical values of $u^2$ are
\begin{equation}\label{eq:physical-a}
a_-=\frac{13\eta+6}{4+8\eta},\qquad
a_0=\frac32\eta,\qquad
a_+=\frac{13\eta-6}{4-8\eta}.
\end{equation}
We use $a_e=a_-,a_0,a_+$ when the energy is indexed by $e=-4,0,4$.

\subsection{A formal moment deformation}

We keep the recurrence \eqref{eq:moment-recurrence} but replace the initial value $\mu(u^2)=\eta$ by a real parameter $t$. More precisely, define $\mu_t$ by
\begin{equation}\label{eq:formal-moments}
\mu_t(1)=1,\quad \mu_t(u^2)=t,\quad
\mu_t(u^{2j+1})=0,\quad
\mu_t(u^{n+4})=\frac{n+1}{2}\mu_t(u^n).
\end{equation}
All its moments are affine in $t$, and $\mu_\eta=\mu$. For every polynomial $f$,
\begin{equation}\label{eq:formal-ibp}
\mu_t(f')=2\mu_t(u^3f).
\end{equation}
Extend $\mu_t$ complex-linearly and put $\langle f,g\rangle_t=\mu_t(\overline f g)$. Multiplication by $u$ and the gauged momentum
\[
\widetilde p_u=-\mathrm i(\partial_u-u^3)
\]
are formally symmetric for this pairing. Define $\widetilde A$ by $A(f\phi)=(\widetilde A f)\phi$. By composition,
\begin{equation}\label{eq:formal-adjoint}
\langle f,\widetilde A g\rangle_t
=\langle\widetilde{A^\dagger}f,g\rangle_t
\qquad(A\in\mathcal A_1).
\end{equation}
This follows algebraically from \eqref{eq:formal-ibp} and holds for every real $t$.

For $t$ near $\eta$, define a functional on the full Weyl algebra by
\begin{equation}\label{eq:oneD-functional}
\omega_{e,t}(A)=\frac{\mu_t(P_e\widetilde A P_e)}{\mu_t(P_e^2)}.
\end{equation}
The denominators are $4+8t$, $2$ and $4-8t$, so all are positive in a neighborhood of $\eta$. Equations~\eqref{eq:qes-states} and \eqref{eq:formal-adjoint} imply Hermiticity and
\begin{equation}\label{eq:oneD-identities}
\omega_{e,t}(A(h_u-e))
=\omega_{e,t}((h_u-e)A)
=\omega_{e,t}([h_u,A])=0
\qquad(A\in\mathcal A_1).
\end{equation}
They also give
\begin{equation}\label{eq:oneD-square}
\omega_{e,t}(A^\dagger A)
=\frac{\mu_t(|\widetilde A P_e|^2)}{\mu_t(P_e^2)}.
\end{equation}
If $\deg A\le K$, then $\deg(\widetilde A P_e)\le3K+4$: each position factor raises degree by one and each gauged momentum factor by at most three. Define
\begin{equation}\label{eq:JDt}
J_D(t)=\bigl(\mu_t(u^{a+b})\bigr)_{a,b=0}^D.
\end{equation}
The strictly positive density $|\phi|^2$ makes $J_D(\eta)$ positive definite for every finite $D$ \cite{Gautschi2004,Schmudgen2017}. By continuity,
\begin{equation}\label{eq:oneD-Hankel}
J_{3K+4}(t)\succ0
\quad\Longrightarrow\quad M_K(\omega_{e,t})\succeq0
\end{equation}
and, since $J_{3K+4}(t)$ is affine in $t$, its hypothesis holds throughout a neighborhood of $\eta$ that may depend on $K$.

The observable changes along the deformation:
\begin{align}
a_-(t):=\omega_{-4,t}(u^2)&=\frac{13t+6}{4+8t},&
a_0(t):=\omega_{0,t}(u^2)&=\frac32t,&
a_+(t):=\omega_{4,t}(u^2)&=\frac{13t-6}{4-8t}.\label{eq:a-deformed}
\end{align}
Their derivatives are $4/(4+8t)^2$, $3/2$ and $4/(4-8t)^2$, respectively. All are positive near $\eta$.

\subsection{Finite constraints and finite certificates}

Let
\begin{equation}
\mathcal J_e=\{A\in\mathcal A_1:A\psi_e=0\}
\end{equation}
be the annihilator of the eigenstate, a left ideal of $\mathcal{A}_1$ strictly
containing $\mathcal{A}_1(h_u-e)$. A finite relaxation may additionally impose
$L(A)=0$ for every $A\in\mathcal J_e+\mathcal J_e^\dagger$ lying in its retained space,
including elements obtainable only as combinations of higher-degree elements
whose top-degree terms cancel.

\begin{theorem}[Persistence under finite polynomial tests]\label{thm:oneD-persistence}
Fix $e\in\{-4,0,4\}$. Let $\mathcal W\subset\mathcal A_1$ be a finite-dimensional $*$-closed linear space containing $1$ and $u^2$, and let $G_1,\ldots,G_m$, $m\ge1$, be Weyl polynomials with $G_i^\dagger G_j\in\mathcal W$. There are normalized Hermitian functionals $L^\pm:\mathcal W\to\mathbb C$ such that
\[
\bigl(L^\pm(G_i^\dagger G_j)\bigr)_{i,j}\succeq0,\qquad
L^\pm\bigl(\mathcal W\cap(\mathcal J_e+\mathcal J_e^\dagger)\bigr)=0,
\]
and
\[
L^-(u^2)<a_e<L^+(u^2).
\]
In particular, if $\mathcal F_K^{h}(e)$ denotes \eqref{eq:finite-hierarchy} for $h_u$, then for every $K\ge3$,
\begin{equation}\label{eq:oneD-strict}
\inf_{L\in\mathcal F_K^{h}(e)}L(u^2)<a_e<
\sup_{L\in\mathcal F_K^{h}(e)}L(u^2).
\end{equation}
\end{theorem}

\begin{proof}
Set $D=\max_i\deg G_i$. Equations~\eqref{eq:oneD-square}--\eqref{eq:oneD-Hankel} give the required Gram positivity for $\omega_{e,t}$ when $t$ is sufficiently close to $\eta$.
If $A\in\mathcal J_e$, then $A(P_e\phi)=0$. Since $\phi$ never vanishes, $\widetilde A P_e=0$ is a polynomial identity, independent of $t$. Hence $\omega_{e,t}(A)=0$ for all such $t$, and Hermiticity gives $\omega_{e,t}(A^\dagger)=0$ as well. The full functional therefore vanishes on $\mathcal J_e+\mathcal J_e^\dagger$, including every element of its intersection with $\mathcal W$.
Take $L^\pm=\omega_{e,t^\pm}|_{\mathcal W}$ with $t^-<\eta<t^+$ in the positivity neighborhood. Equation~\eqref{eq:a-deformed} gives the strict inequalities. For the standard hierarchy, the identities \eqref{eq:oneD-identities} supply every retained eigenstate and stationarity relation.
\end{proof}

The bounds in \eqref{eq:oneD-strict} are finite. Indeed,
\[
u^6-12u^2+16=(u^3-2u)^2+4(u^2-2)^2
\]
and the mean-energy relation imply
\begin{equation}\label{eq:oneD-bounded}
0\le L(u^2)\le2e+16\qquad(L\in\mathcal F_K^{h}(e),\ K\ge3).
\end{equation}

Let $\Sigma_{\mathcal A_1}$ denote the cone of finite sums $\sum_j A_j^\dagger A_j$. The theorem has a direct consequence for certificates of the physical value.

\begin{corollary}[No finite certificate of the physical value]\label{cor:oneD-no-certificate}
For each $e\in\{-4,0,4\}$,
\begin{equation}\label{eq:oneD-no-certificate}
u^2-a_e\notin\Sigma_{\mathcal A_1}+\mathcal J_e+\mathcal J_e^\dagger,
\qquad
a_e-u^2\notin\Sigma_{\mathcal A_1}+\mathcal J_e+\mathcal J_e^\dagger.
\end{equation}
\end{corollary}

\begin{proof}
A proposed certificate uses finitely many squares and annihilator terms. Let the square factors be $A_1,\ldots,A_m$, setting $A_1=1$ if the sum is empty. Choose a finite-dimensional $*$-closed space $\mathcal W$ containing $1$, $u^2$, every annihilator term and every cross product $A_i^\dagger A_j$. Apply Theorem~\ref{thm:oneD-persistence} with $G_i=A_i$. The appropriate functional $L^-$ or $L^+$ is nonnegative on the right-hand side of the proposed identity and negative on the left-hand side.
\end{proof}

Stationarity is already included in these annihilator relations, since
\begin{equation}\label{eq:stationarity-ideal}
[h_u,O]=(h_u-e)O-O(h_u-e)\in\mathcal J_e^\dagger+\mathcal J_e.
\end{equation}
The obstruction concerns finite polynomial-square certificates. Strict positivity theorems for the Weyl algebra can allow multipliers and different certificate classes \cite{Schmudgen2005}.

\subsection{Finite localizing constraints}\label{Sec:Localizing}

For a known nonnegative observable $g$, a localizing test imposes positivity of a weighted Gram matrix $\bigl(L(G_i^\dagger gG_j)\bigr)_{i,j}$. Say that a Hermitian Weyl polynomial $g$ is \emph{strictly positive on polynomial vectors} if
\begin{equation}\label{eq:strict-weight}
\langle f\phi,gf\phi\rangle_{L^2(\mathbb R)}>0
\quad\text{for every nonzero }f\in\mathbb C[u].
\end{equation}
Examples include every nonzero nonnegative position polynomial $g(u)$, every nonzero nonnegative momentum polynomial $g(p_u)$, and $h_u+c$ for $c>4$. For momentum polynomials this follows from Plancherel's theorem and the fact that a nonzero polynomial has a zero set of measure zero.

\begin{proposition}[Persistence under finite localizing tests]\label{prop:localizing}
Theorem~\ref{thm:oneD-persistence} remains valid after adding finitely many tests
\[
\bigl(L(G_{k,i}^\dagger g_kG_{k,j})\bigr)_{i,j}\succeq0,\qquad k=1,\ldots,s,
\]
provided each $g_k$ satisfies \eqref{eq:strict-weight} and each displayed product lies in $\mathcal W$. Consequently, \eqref{eq:oneD-no-certificate} remains true when $\Sigma_{\mathcal A_1}$ is enlarged by the cones of finite sums $\sum_i A_i^\dagger g_kA_i$ for any fixed finite collection of such weights.
\end{proposition}

\begin{proof}
Fix $k$ and let $S_k$ be the image of $\operatorname{span}\{G_{k,i}\}_i$ under $G\mapsto\widetilde G P_e$. On this finite-dimensional polynomial space consider
\[
q_{k,t}(f)=\mu_t(\overline f\,\widetilde g_k f).
\]
Its matrix entries depend affinely on $t$, and the formal adjoint identity makes it Hermitian. At $t=\eta$ it is positive definite by \eqref{eq:strict-weight}. Hence it stays positive definite near $\eta$. Moreover,
\[
\omega_{e,t}(G^\dagger g_kG)
=\frac{q_{k,t}(\widetilde G P_e)}{\mu_t(P_e^2)}.
\]
This proves the weighted Gram inequality, including any kernel arising from $\widetilde G P_e=0$. Intersect the finitely many positivity neighborhoods and choose $t$ on either side of $\eta$. The certificate statement follows by the same evaluation argument as in Corollary~\ref{cor:oneD-no-certificate}.
\end{proof}

The argument requires $q_{k,\eta}$ to be positive definite on $S_k$: a form
that is only positive semidefinite at $t=\eta$ can become indefinite under
arbitrarily small changes of $t$. The admissible neighborhood of $\eta$ depends
on the finite collection of tests. Thus each finite relaxation has feasible
parameters on both sides of $\eta$, but no single $t\neq\eta$ is asserted to be
feasible at all degrees.

\section{Convergence to the physical energy shell}\label{Sec:Convergence}

Under the polynomial hypotheses below, the eigenstate relations give moment bounds uniform in the truncation order. These bounds yield reconstruction and convergence to the physical energy shell.

\subsection{Polynomial hypotheses}

Consider $H_V$ and $\mathcal F_K^V(E)$ from Section~\ref{Sec:Hierarchy}, with $V\in\mathbb R[q_1,\ldots,q_d]$ of degree $2r$. Assume that $H_V$ is essentially self-adjoint on $C_c^\infty(\mathbb R^d)$, and write $\partial_j=\partial/\partial q_j$. Let $\Sigma_{2r}$ be the cone of real sums of squares $\sum_s g_s(q)^2$ with $\deg g_s\le r$.

We assume that there are real constants $v_p,v_q$ (not necessarily equal), nonnegative constants $a_{j\ell},b_{j\ell}$ for $1\le j\le d$ and $1\le\ell\le r$, positive constants $\epsilon_j$, and integers $1\le r_j\le r$ such that
\begin{align}
V-v_p&\in\Sigma_{2r},\label{eq:sos-potential}\\
a_{j\ell}(V-v_p)+b_{j\ell}-(-1)^{\ell+1}\partial_j^{2\ell}V&\in\Sigma_{2r},\qquad 1\le\ell\le r,\label{eq:sos-derivative}\\
V-v_q-\sum_{j=1}^d\epsilon_jq_j^{2r_j}&\in\Sigma_{2r}.\label{eq:sos-confinement}
\end{align}
Each condition is a fixed polynomial identity that can be checked by a standard 
semidefinite program. Its role here is that, being a sum of squares, it is
evaluated nonnegatively by every truncated positive functional, so the
resulting inequalities hold on the whole feasible set and not only for
physical states, as in moment--SOS hierarchies \cite{Parrilo2003,Putinar1993}.

A simple coupled example is
\begin{equation}\label{eq:coupled-quartic}
V(x,y)=x^4+y^4+x^2y^2.
\end{equation}
This potential is a sum of squares, and
\begin{align*}
V-x^4-y^4&=(xy)^2,\\
V+37-\partial_x^2V&=(x^2-6)^2+(y^2-1)^2+(xy)^2.
\end{align*}
The second identity also holds with $x$ and $y$ interchanged. Since $\partial_x^4V=\partial_y^4V=24$, the hypotheses hold with
\[
v_p=v_q=0,\qquad r=r_x=r_y=2,\qquad \epsilon_x=\epsilon_y=1,
\]
and $(a_{j1},b_{j1})=(1,37)$, $(a_{j2},b_{j2})=(0,0)$ for $j=x,y$. The smooth nonnegative potential also gives essential self-adjointness \cite{Simon1973}.

Navascu\'es et al. reconstruct a Schr\"odinger state from a positive Weyl functional satisfying
\[
|L(s)|\le C_0c_0^{|s|}\Gamma\!\left(\frac{|s|+k_0+1}{2}\right)
\]
for all creation and annihilation monomials $s$, with fixed constants $C_0,c_0>0$ and $k_0\in\mathbb N$ \cite[Lemma~7]{Navascues2013}. Here the eigenstate relations and \eqref{eq:sos-potential}--\eqref{eq:sos-confinement} give moment bounds uniform over the finite feasible sets. The momentum bound \eqref{eq:p-growth-finite} allows growth of order $D_j^n(n!)^2$, which exceeds the growth allowed by that lemma. We establish regularity through analytic vectors in the GNS representation, then use the eigenstate relations to identify the energy support. This yields convergence of every fixed-degree moment projection to the physical moment set at the prescribed energy.

\subsection{Uniform moment bounds}

The momentum estimate rests on the identity
\begin{equation}\label{eq:sandwich-identity}
\frac12\{p_j^{2m},f\}=\sum_{\ell=0}^{m}(-1)^\ell c_{m\ell}\,p_j^{m-\ell}(\partial_j^{2\ell}f)p_j^{m-\ell},
\end{equation}
valid for every real polynomial $f$ and $m\ge1$, where $\{A,B\}:=AB+BA$ denotes the anticommutator and
\begin{equation}\label{eq:sandwich-coeff}
c_{m\ell}=\frac{m}{m+\ell}\binom{m+\ell}{2\ell}.
\end{equation}
It follows, for example, by writing left and right multiplication by $p_j$ as $X$ and $Y$ and expanding $(X^{2m}+Y^{2m})/2$ in the variables $XY$ and $(X-Y)^2$, with $(X-Y)f=[p_j,f]=-\mathrm{i}\partial_jf$.

\begin{lemma}[Uniform moment bounds]\label{lem:uniform-moment-bounds}
Assume \eqref{eq:sos-potential}--\eqref{eq:sos-confinement} and fix $E\in\mathbb{R}$. For each $j$ there exist constants $C_j,D_j\ge1$, independent of $K$, such that for every $K\ge r$ with $\mathcal F_K^V(E)\neq\varnothing$ and every $L\in\mathcal F_K^V(E)$,
\begin{align}
0\le L(p_j^{2n})&\le D_j^n(n!)^2,\qquad 0\le n\le K-r+1,\label{eq:p-growth-finite}\\
0\le L(q_j^{2n})&\le C_j^n(n!)^{2/(r_j+1)},\qquad 0\le n\le K-r+r_j.\label{eq:q-growth-finite}
\end{align}
In particular, every fixed PBW moment is uniformly bounded once $K$ is sufficiently large.
\end{lemma}

\begin{proof}
Set
\[
U=V-v_p,\qquad \Delta_p=E-v_p,\qquad P_m=L(p_j^{2m}),\qquad R_m=L(p_j^mUp_j^m).
\]
Write $U=\sum_s g_s(q)^2$ with $\deg g_s\le r$. Whenever $m+r\le K$,
\[
p_j^mUp_j^m=\sum_s(g_sp_j^m)^\dagger(g_sp_j^m),
\]
so truncated positivity gives $P_m,R_m\ge0$ whenever these quantities are visible. The mean-energy relation and \eqref{eq:sos-potential} also give $\Delta_p\ge0$ whenever the feasible set is nonempty. Taking the real part of the eigenstate relation with multiplier $p_j^{2m}$, using \eqref{eq:sandwich-identity}, and dropping the nonnegative kinetic terms in the other directions gives, for $m+r\le K$,
\begin{equation}\label{eq:moment-recursion}
\frac12P_{m+1}+R_m
\le
\Delta_pP_m+
\sum_{\ell=1}^{\min(m,r)}
c_{m\ell}\bigl(a_{j\ell}R_{m-\ell}+b_{j\ell}P_{m-\ell}\bigr).
\end{equation}
The estimate \eqref{eq:sos-derivative} is used here after sandwiching it between $p_j^{m-\ell}$. At $m=0$, the mean-energy relation gives $P_1/2+R_0\le\Delta_p$. Define $B_0=1$, $B_1=2\Delta_p$, and
\[
B_{m+1}
=
2\Delta_pB_m+
\sum_{\ell=1}^{\min(m,r)}
c_{m\ell}\bigl(a_{j\ell}B_{m-\ell+1}+2b_{j\ell}B_{m-\ell}\bigr).
\]
Simultaneous induction gives $P_m\le B_m$ and $R_m\le B_{m+1}/2$. To make the factorial bound explicit, define
\begin{equation}\label{eq:explicit-D}
\kappa_1=\frac12,\qquad
\kappa_\ell=\frac{\ell^{\ell-2}}{(2\ell)!}\quad(\ell\ge2),\qquad
D_j=1+2\Delta_p+\sum_{\ell=1}^{r}\kappa_\ell(a_{j\ell}+2b_{j\ell}).
\end{equation}
For $m\ge\ell\ge1$, the coefficient formula gives
\begin{equation}\label{eq:coefficient-bound}
c_{m\ell}\left(\frac{(m-\ell+1)!}{(m+1)!}\right)^2\le\kappa_\ell.
\end{equation}
For $\ell=1$ the left-hand side is $m^2/[2(m+1)^2]$. For $\ell\ge2$ it equals
\[
\frac{1}{(2\ell)!}
\frac{m(m-\ell+1)\prod_{i=2}^{\ell-1}(m+i)}
{\prod_{i=m-\ell+2}^{m+1}i}.
\]
Pair $m$ with $m+1$, and $m-\ell+1$ with $m-\ell+2$; both ratios are at most one. Each remaining ratio $(m+i)/(m-\ell+1+i)$ is at most $\ell$. This proves \eqref{eq:coefficient-bound}, including the empty product for $\ell=2$.
Assuming $B_n\le D_j^n(n!)^2$ for the preceding indices and dividing the recursion by $D_j^{m+1}((m+1)!)^2$, the sum of the resulting coefficients is at most
\[
\frac{2\Delta_p}{D_j}
+\sum_{\ell=1}^{r}\kappa_\ell
\left(a_{j\ell}D_j^{-\ell}+2b_{j\ell}D_j^{-\ell-1}\right)\le\frac{D_j-1}{D_j}<1.
\]
By induction, $B_m\le D_j^m(m!)^2$, proving \eqref{eq:p-growth-finite}.

For the position moments write
\[
Q_a=L(q_j^{2a}),\qquad \Delta_q=E-v_q.
\]
The mean-energy relation and \eqref{eq:sos-confinement} imply $\Delta_q\ge0$.
The identity
\[
\frac12\{q_j^{2a},p_j^2\}=q_j^ap_j^2q_j^a-a^2q_j^{2a-2},
\]
together with \eqref{eq:sos-confinement}, gives
\begin{equation}\label{eq:q-recursion}
\epsilon_jQ_{a+r_j}\le\Delta_qQ_a+\frac{a^2}{2}Q_{a-1},
\qquad a+r\le K,
\end{equation}
where the last term is omitted for $a=0$. The mean-energy relation gives $\epsilon_jQ_{r_j}\le\Delta_q$. For $1\le s\le r_j$, the univariate polynomial $1+x^{2r_j}-x^{2s}$ is nonnegative on $\mathbb R$ and hence is a sum of squares; evaluating it with $L$ gives
\[
Q_s\le1+\frac{\Delta_q}{\epsilon_j}.
\]
Set
\[
C_j=2+\frac{2\Delta_q+1}{\epsilon_j}.
\]
For $n=a+r_j$, substitution of the induction hypothesis into \eqref{eq:q-recursion} shows that the two normalized contributions are bounded respectively by
\[
\frac{\Delta_q}{\epsilon_j}C_j^{-r_j},\qquad
\frac{1}{2\epsilon_j}C_j^{-(r_j+1)},
\]
where we used
\[
a^2\left(\frac{(a-1)!}{(a+r_j)!}\right)^{2/(r_j+1)}\le1.
\]
Since $r_j\ge1$ and $C_j\ge1$, their sum is at most
\[
\frac{\Delta_q+1/2}{\epsilon_jC_j}
=\frac{C_j-2}{2C_j}<1.
\]
For $1\le s\le r_j$, the initial bound is at most $C_j$, and therefore at most $C_j^s(s!)^{2/(r_j+1)}$. This closes the induction and proves \eqref{eq:q-growth-finite}.

For multi-indices $\alpha,\beta\in\mathbb N_0^d$, write $q^\alpha=\prod_jq_j^{\alpha_j}$, $p^\beta=\prod_jp_j^{\beta_j}$, $|\alpha|=\sum_j\alpha_j$, and $\alpha!:=\prod_j\alpha_j!$, with analogous notation for $\beta$. Repeated Cauchy--Schwarz bounds every fixed multivariable position or momentum moment by finitely many one-coordinate even moments. Applying Cauchy--Schwarz once more to the PBW-ordered monomial $q^\alpha p^\beta$ proves the final assertion.
\end{proof}

\subsection{Reconstruction and convergence}

Let $\pi_D$ denote restriction to PBW moments of degree at most $D$, and let $\mathcal Q_D^V(E)$ be the set of such moment vectors obtained from density operators supported on $\ker(H_V-E)$. Distances between these truncated moment vectors are computed in the finite-dimensional real coordinate space obtained by splitting complex moments into real and imaginary parts; the particular choice of norm is immaterial for convergence. Condition \eqref{eq:sos-confinement} implies
\[
V(q)\ge c|q|^2-C
\]
for suitable $c>0$ and $C$, so $H_V$ has compact resolvent.

Its eigenfunctions have all polynomial Weyl moments finite. To see this, let $H_V\psi=E\psi$ and write $\langle q\rangle=(1+|q|^2)^{1/2}$. For any $a>0$,
\[
V-E-\frac12|\nabla(a\langle q\rangle)|^2\ge1
\]
outside a compact set. The Agmon estimate, applied first to bounded truncations of $a\langle q\rangle$, gives $e^{a\langle q\rangle}\psi\in L^2$ \cite{Agmon1982}. To control derivatives, iterate interior elliptic estimates for $-\Delta\psi=2(E-V)\psi$ on concentric balls. Since every derivative of $V$ has polynomial growth, for each integer $s\ge0$ there are constants $C_s,N_s$ independent of the center $z$ such that
\[
\|\psi\|_{H^{s+2}(B_1(z))}
\le C_s\langle z\rangle^{N_s}\|\psi\|_{L^2(B_2(z))}.
\]
Summing over a cover by unit balls whose doubled balls have uniformly bounded overlap, the exponential bound gives $\langle q\rangle^N\partial^\beta\psi\in L^2$ for every $N\ge0$ and multi-index $\beta$. Every polynomial in $q,p$ therefore acts on $\psi$ with finite norm. Each eigenspace is finite dimensional, so its density operators form a compact set and their polynomial moments depend continuously on the density operator. Hence $\mathcal Q_D^V(E)$ is compact.

\begin{theorem}[Reconstruction and convergence]\label{thm:reconstruction}
Assume that $H_V$ is essentially self-adjoint on $C_c^\infty(\mathbb R^d)$ and that \eqref{eq:sos-potential}--\eqref{eq:sos-confinement} hold. Let $K_n\to\infty$ and choose $L_n\in\mathcal F_{K_n}^V(E)$. Then a subsequence converges moment by moment on the full polynomial Weyl algebra to a functional
\begin{equation}\label{eq:rho-representation}
L(A)=\operatorname{Tr}(\rho A),
\end{equation}
where $\rho$ is a positive trace-class operator of unit trace in the Schr\"odinger representation and
\begin{equation}\label{eq:energy-support}
\rho=P_E\rho P_E,
\end{equation}
with $P_E$ the spectral projection onto $\ker(H_V-E)$.

If $E\in\operatorname{spec}(H_V)$, then for every fixed $D$,
\begin{equation}\label{eq:set-convergence}
\lim_{K\to\infty}
\sup_{L\in\mathcal F_K^V(E)}
\operatorname{dist}\bigl(\pi_DL,\mathcal Q_D^V(E)\bigr)=0.
\end{equation}
If $E\notin\operatorname{spec}(H_V)$, then $\mathcal F_K^V(E)$ is empty for all sufficiently large $K$. The convergence statement remains valid for any strengthened feasible sets contained in $\mathcal F_K^V(E)$ that retain the truncations of all physical states at energy $E$. No quantitative convergence rate or effective exclusion level is asserted.
\end{theorem}

\begin{proof}
Enumerate the PBW monomials. Lemma~\ref{lem:uniform-moment-bounds} and a diagonal subsequence argument give pointwise convergence of all moments to a linear functional $L$ on the full Weyl algebra. Normalization and Hermiticity pass to the limit. Every fixed square and every fixed eigenstate or stationarity relation is imposed at all sufficiently high truncation orders, and pointwise limits preserve these inequalities and identities; hence
\begin{equation}\label{eq:infinite-positive}
L(1)=1,\qquad L(A^\dagger A)\ge0,
\end{equation}
and the full polynomial eigenstate relations hold.

The restrictions of $L$ to the commuting position and momentum polynomial algebras satisfy the Carleman condition. Indeed,
\[
L(p_j^{2n})^{-1/(2n)}
\ge D_j^{-1/2}(n!)^{-1/n}
\ge D_j^{-1/2}n^{-1},
\]
while
\[
L(q_j^{2n})^{-1/(2n)}
\ge C_j^{-1/2}(n!)^{-1/[n(r_j+1)]}
\ge C_j^{-1/2}n^{-1/(r_j+1)}.
\]
Both corresponding series diverge, with the usual convention that a zero even moment contributes $+\infty$. The restrictions of $L$ to $\mathbb R[q_1,\ldots,q_d]$ and $\mathbb R[p_1,\ldots,p_d]$ are normalized positive functionals on commutative polynomial algebras. The coordinatewise Carleman conditions therefore allow the multivariate Carleman--Nussbaum theorem to be applied to each restriction, yielding representing probability measures $\mu_q$ and $\mu_p$ \cite[Proposition~4.8]{Ghasemi2016}; see also Refs.~\cite{Nussbaum1965,deJeu2003,Schmudgen2017}. These marginal measures allow us to use classical inequalities within each commuting algebra. For $n=|\alpha|$, H\"older's inequality in $\mu_q$ gives
\[
L(q^{2\alpha})\le\prod_{j:\alpha_j>0}L(q_j^{2n})^{\alpha_j/n},
\]
and the same argument applies to momentum. Combining the one-coordinate estimates with $n!/\alpha!\le d^n$, we may take $M=d^2\max_j\{C_j,D_j\}\ge1$ to obtain
\begin{equation}\label{eq:marginal-growth}
L(q^{2\alpha})\le M^{|\alpha|}(\alpha!)^2,\qquad
L(p^{2\beta})\le M^{|\beta|}(\beta!)^2.
\end{equation}
Cauchy--Schwarz for the full positive functional then yields
\begin{equation}\label{eq:normal-growth}
|L(q^\alpha p^\beta)|
\le
M^{(|\alpha|+|\beta|)/2}\alpha!\,\beta!.
\end{equation}

Apply the GNS construction to $L$ \cite{Schmudgen1990}. For completeness, let
\[
\mathcal N_L:=\{A\in\mathcal A_d:L(A^\dagger A)=0\}.
\]
Cauchy--Schwarz for the full positive functional gives $L(C^\dagger A)=0$ for every $C$ whenever $A\in\mathcal N_L$. Taking $C=B^\dagger BA$ shows $L((BA)^\dagger BA)=0$, so $\mathcal N_L$ is a left ideal. Left multiplication is therefore well defined on the quotient, and after completion we obtain the cyclic representation $(\pi_L,\mathcal H_L,\Omega)$. 
The vectors $\pi_L(A)\Omega$, $A\in\mathcal A_d$, form a subspace $\mathcal D_L$, dense by construction of the completion and invariant under every $\pi_L(B)$, since $\pi_L(B)\pi_L(A)\Omega=\pi_L(BA)\Omega$. We suppress $\pi_L$ in the notation below; left multiplication by the canonical generators defines symmetric operators $q_j,p_j$ on $\mathcal D_L$.

The norm growth needed for analyticity follows explicitly from \eqref{eq:normal-growth}. The multi-index normal-ordering identity
\[
p^\beta q^{2\alpha}
=
\sum_{\substack{\gamma\le\beta\\ \gamma\le2\alpha}}
\binom{\beta}{\gamma}(-\mathrm{i})^{|\gamma|}
\frac{(2\alpha)!}{(2\alpha-\gamma)!}
q^{2\alpha-\gamma}p^{\beta-\gamma}
\]
together with $\|q^\alpha p^\beta\Omega\|^2=L(p^\beta q^{2\alpha}p^\beta)$ and \eqref{eq:normal-growth} gives
\[
\|q^\alpha p^\beta\Omega\|^2
\le e^{d/M}M^{|\alpha|+|\beta|}(2\alpha)!(2\beta)!.
\]
Indeed, after dividing each normal-ordered term by $M^{|\alpha|+|\beta|}(2\alpha)!(2\beta)!$, its absolute value is at most $M^{-|\gamma|}/\gamma!$, since
\[
\binom{\beta}{\gamma}\frac{(2\beta-\gamma)!}{(2\beta)!}\le\frac1{\gamma!}.
\]
Summing over $\gamma\in\mathbb N_0^d$ gives $e^{d/M}$. Factorials and inequalities between multi-indices are understood componentwise. Since $(2n)!\le4^n(n!)^2$,
\begin{equation}\label{eq:vector-growth}
\|q^\alpha p^\beta\Omega\|
\le
e^{d/(2M)}(2\sqrt M)^{|\alpha|+|\beta|}\alpha!\,\beta!.
\end{equation}
For a fixed monomial $A=q^\mu p^\nu$, normal ordering gives the more explicit formula
\[
q_j^n p_k^m A
=\sum_{c=0}^{\min(m,\mu_k)}
\binom mc(-\mathrm i)^c\frac{\mu_k!}{(\mu_k-c)!}
q^{\mu+n\mathbf e_j-c\mathbf e_k}
p^{\nu+(m-c)\mathbf e_k},
\]
where $\mathbf e_j$ is the $j$th coordinate vector. For fixed $a\ge0$,
\[
(n+a)!\le a!\,n!(n+1)^a,
\qquad
\binom mc\frac{\mu_k!}{(\mu_k-c)!}\le\binom{\mu_k}{c}m^c.
\]
Thus the position and momentum factorials are bounded by $\mu!\,n!(n+1)^{|\mu|}$ and $\nu!\,m!(m+1)^{|\nu|}$, while the sum of the coefficient bounds is at most $(m+1)^{\mu_k}$. Applying \eqref{eq:vector-growth} term by term gives the explicit constants
\[
R=2\sqrt M,\qquad
C_A=e^{d/(2M)}R^{|\mu|+|\nu|}\mu!\,\nu!,\qquad
N_A=2|\mu|+|\nu|,
\]
and the estimate
\begin{equation}\label{eq:analytic-vectors}
\|q_j^n p_k^mA\Omega\|
\le C_AR^{n+m}n!\,m!(n+m+1)^{N_A}
\qquad(n,m\ge0).
\end{equation}
Finite sums are handled by enlarging the constants. Setting $m=0$ or $n=0$ shows that, for $0<s<R^{-1}$,
\[
\sum_{n=0}^\infty\frac{s^n}{n!}\|q_j^nA\Omega\|<\infty,
\qquad
\sum_{n=0}^\infty\frac{s^n}{n!}\|p_j^nA\Omega\|<\infty.
\]
Every vector in $\mathcal D_L$ is therefore analytic for every canonical generator in the sense of Nelson \cite{Nelson1959}.

We can now integrate the infinitesimal canonical relations. Represent the Heisenberg Lie algebra on $\mathcal D_L$ by
\[
X_j\mapsto\mathrm{i}q_j,\qquad
Y_j\mapsto\mathrm{i}p_j,\qquad
Z\mapsto-\mathrm{i}I.
\]
The maps above form a representation of the finite-dimensional Heisenberg Lie algebra by skew-symmetric operators on the common dense invariant domain $\mathcal D_L$. Equation~\eqref{eq:analytic-vectors} shows that every vector in this domain is analytic for each basis generator, while the central generator is bounded. Hypothesis~(C) and Theorem~1 of Flato, Simon, Snellman and Sternheimer therefore apply \cite{Flato1972}: the infinitesimal representation integrates to a strongly continuous unitary representation of the simply connected Heisenberg group. Analyticity makes each canonical generator essentially self-adjoint by Nelson's theorem. The integrated one-parameter groups therefore have generators equal to these closures $Q_j,P_j$, and the group law gives the regular Weyl relations
\begin{equation}\label{eq:Weyl-relations}
e^{\mathrm{i}sQ_j}e^{\mathrm{i}tP_k}
=
e^{-\mathrm{i}st\delta_{jk}}
e^{\mathrm{i}tP_k}e^{\mathrm{i}sQ_j}.
\end{equation}
The Stone--von Neumann theorem with multiplicity gives, for a Hilbert space $\mathcal K$ (the multiplicity space), \cite[Theorem~2.2]{Rosenberg2004}
\begin{equation}\label{eq:SvN-decomposition}
\mathcal H_L\simeq L^2(\mathbb R^d)\otimes\mathcal K.
\end{equation}
Under this unitary equivalence the Weyl groups are the standard Schr\"odinger Weyl groups on the first factor. Uniqueness of self-adjoint generators therefore identifies $Q_j$ with multiplication by $q_j$ and $P_j$ with $-\mathrm{i}\partial_{q_j}$, tensored with the identity on $\mathcal K$.

Let $\Psi$ be the image of $\Omega$ under this identification and define
\[
\rho=\operatorname{Tr}_{\mathcal K}|\Psi\rangle\langle\Psi|,
\]
where $\operatorname{Tr}_{\mathcal K}$ denotes the partial trace over the multiplicity space $\mathcal K$. Then $\rho\ge0$, $\operatorname{Tr}\rho=1$, and $\rho$ is trace class. Writing
\[
\Psi=\sum_\ell\psi_\ell\otimes e_\ell
\]
in an orthonormal basis of $\mathcal K$, note that $\Omega\in\mathcal D_L$ and $\mathcal D_L$ is invariant under the polynomial algebra. Hence $\Psi$ lies in the domain of $A\otimes I$ for every polynomial $A$, and
\[
\sum_\ell\|A\psi_\ell\|^2
=
\|(A\otimes I)\Psi\|^2
<\infty.
\]
Consequently, the series
\[
A\rho=\sum_\ell|A\psi_\ell\rangle\langle\psi_\ell|
\]
converges absolutely in trace norm, since
\[
\sum_\ell\|A\psi_\ell\|\,\|\psi_\ell\|
\le
\left(\sum_\ell\|A\psi_\ell\|^2\right)^{1/2}
\left(\sum_\ell\|\psi_\ell\|^2\right)^{1/2}
<\infty.
\]
Taking its trace gives
\[
L(A)=\operatorname{Tr}(A\rho)
=\sum_\ell\langle\psi_\ell,A\psi_\ell\rangle.
\]
Applying the same argument to $A^\dagger$ shows that $\rho A$, initially defined on $\operatorname{Dom}(A)$, extends uniquely to the trace-class operator $(A^\dagger\rho)^\dagger$. This extension has trace $\operatorname{Tr}(A\rho)$, which justifies the notation $\operatorname{Tr}(\rho A)$ in \eqref{eq:rho-representation}.

It remains to identify the energy support. Taking $A=p_j^2$ and $A=V$ in the square-summability relation above shows that every component $\psi_\ell$ lies in the domains of $p_j^2$ and of multiplication by $V$, with the corresponding squared norms summable over $\ell$. Hence $(-\Delta/2+V)\psi_\ell\in L^2(\mathbb R^d)$ in the distributional sense. Let $H_{\min}$ denote $-\Delta/2+V$ on $C_c^\infty(\mathbb R^d)$. For every $f\in C_c^\infty(\mathbb R^d)$, distributional integration by parts gives
\[
\langle H_{\min}f,\psi_\ell\rangle
=\langle f,(-\Delta/2+V)\psi_\ell\rangle.
\]
Since the second argument on the right belongs to $L^2$, the definition of the adjoint gives $\psi_\ell\in\operatorname{Dom}(H_{\min}^*)$. By the assumed essential self-adjointness, $H_{\min}^*=\overline{H_{\min}}=H_V$, so the distributional expression above is precisely $H_V\psi_\ell$.

The full eigenstate relations allow the multiplier $O=H_V-E$, so
\begin{equation}\label{eq:zero-energy-variance}
0=L((H_V-E)^2)=\|(H_V-E)\Psi\|^2
=\sum_\ell\|(H_V-E)\psi_\ell\|^2.
\end{equation}
Thus every nonzero $\psi_\ell$ lies in $\ker(H_V-E)$, and
\[
\Psi\in\ker(H_V-E)\otimes\mathcal K,
\]
which is equivalent to \eqref{eq:energy-support}.

Suppose now that \eqref{eq:set-convergence} failed for some fixed $D$. Then there would be $K_n\to\infty$ and $L_n\in\mathcal F_{K_n}^V(E)$ whose degree-$D$ projections remain a fixed positive distance from $\mathcal Q_D^V(E)$. The reconstructed subsequence converges instead to the moments of a density operator supported on $\ker(H_V-E)$, a contradiction. If $E$ is not an eigenvalue and arbitrarily high truncation orders were feasible, the same argument would reconstruct a density operator supported on the zero subspace. Hence feasibility must eventually fail. The argument is unchanged for any strengthened hierarchy contained in $\mathcal F_K^V(E)$ and containing all physical truncations.
\end{proof}

\subsection{Application to the sextic oscillator}

For $v(u)=(u^6-11u^2)/2$, set $U_u=v+13$. The identities
\begin{align}
U_u-\frac14u^6
&=\frac14(u^3-2u)^2+\left(u^2-\frac{13}{4}\right)^2+\frac{39}{16},\label{eq:oneD-confinement}\\
40U_u-(15u^4-11)
&=40\left(U_u-\frac14u^6\right)
+10\bigl[u(u^2-1)\bigr]^2+5(u^2-1)^2+6\label{eq:oneD-derivative}
\end{align}
give the required sum-of-squares certificates. In \eqref{eq:sos-potential}--\eqref{eq:sos-confinement}, take
\begin{gather*}
v_p=v_q=-13,\qquad r=r_u=3,\qquad\epsilon_u=\frac14,\\
(a_{u1},b_{u1})=(40,0),\qquad
(a_{u2},b_{u2})=(0,0),\qquad
(a_{u3},b_{u3})=(0,360).
\end{gather*}
The last two choices use $v^{(4)}=180u^2$ and $v^{(6)}=360$. The constants in Lemma~\ref{lem:uniform-moment-bounds} can therefore be chosen as $D_u=50+2e$ and $C_u=110+8e$ at each of the energies $e=-4,0,4$.

Define
\[
\underline b_K^{(e)}=\inf_{L\in\mathcal F_K^h(e)}L(u^2),\qquad
\overline b_K^{(e)}=\sup_{L\in\mathcal F_K^h(e)}L(u^2).
\]
Restriction from level $K+1$ to level $K$ preserves feasibility, so these bounds are monotone. Theorems~\ref{thm:oneD-persistence} and~\ref{thm:reconstruction} give
\begin{equation}\label{eq:oneD-limits}
\underline b_K^{(e)}<a_e<\overline b_K^{(e)}\quad(K\ge3),\qquad
\underline b_K^{(e)}\uparrow a_e,\quad
\overline b_K^{(e)}\downarrow a_e.
\end{equation}
The same one-dimensional hierarchy is therefore nonexact at every finite order and converges to the physical value. The conclusion also holds for nested finite strengthenings of the type in Section~\ref{Sec:OneD} that retain the standard hierarchy.
Table~\ref{tab:bounds} gives numerical estimates of the gaps in \eqref{eq:oneD-limits}.
The reported gaps fall from order one at $K=3$ to $10^{-6}$--$10^{-10}$ at $K=5$;
the largest remaining gaps at $K=5$ occur at $e=4$. Appendix~\ref{app:numerics}
describes the computation.

\begin{table}[htb]
\caption{Numerical estimates of the gaps $a_e-\underline b_K^{(e)}$ and
$\overline b_K^{(e)}-a_e$ (in this order) between the bootstrap bounds on
$\langle u^2\rangle$ and the physical value $a_e$ for the sextic oscillator $h_u$.
See Appendix~\ref{app:numerics} for the computation and solver diagnostics.}
\label{tab:bounds}
\begin{ruledtabular}
\begin{tabular}{lccc}
 & $e=-4$ & $e=0$ & $e=4$\\
$a_e$ & $1.5610933375$ & $0.7169831962$ & $1.2144632215$\\
\hline
$K=3$ & $0.339,\ 0.463$ & $0.504,\ 1.587$ & $1.174,\ 1.270$\\
$K=4$ & $0.0713,\ 3.3\times10^{-5}$ & $0.109,\ 0.508$ & $0.400,\ 0.262$\\
$K=5$ & $1.4\times10^{-9},\ 4.2\times10^{-10}$ & $3.3\times10^{-8},\ 9.1\times10^{-9}$ & $2.8\times10^{-6},\ 7.8\times10^{-7}$\\
\end{tabular}
\end{ruledtabular}
\end{table}

\section{A four-mode example}\label{Sec:Four}

Adding three oscillators to the sextic model allows a second finite-order obstruction: feasible position moments with no positive representing measure. The fixed-observable obstruction and the convergence theorem continue to hold.

\subsection{The exact excited energy shell}

Let $q_u=u$, $q_x=x$, $q_y=y$, $q_z=z$, and set
\begin{equation}\label{eq:mainH}
H=h_u+N,\qquad
N=\frac12\sum_{j=x,y,z}(p_j^2+q_j^2-1).
\end{equation}
The potential is
\begin{equation}\label{eq:main-potential}
V_* =\frac12(u^6-11u^2+x^2+y^2+z^2-3).
\end{equation}
It is smooth and bounded below, so the minimal Schr\"odinger operator is essentially self-adjoint \cite{Simon1973}; confinement gives compact resolvent. Write $\mathcal A=\mathcal A_4$ and $\mathcal F_K(4)=\mathcal F_K^{V_*}(4)$ for the hierarchy of Section~\ref{Sec:Hierarchy}.

Let $\mathcal H_n$ be the eigenspace of $N$ at energy $n\in\mathbb N_0$, with $\dim\mathcal H_n=\binom{n+2}{2}$. The complete shell at $E=4$ is
\begin{equation}\label{eq:exact-shell}
\ker(H-4)
=(\psi_{-4}\otimes\mathcal H_8)
\oplus(\psi_0\otimes\mathcal H_4)
\oplus(\psi_4\otimes\mathcal H_0).
\end{equation}
It has dimension $45+15+1=61$. The proof below establishes completeness by enclosing the two intervening odd levels of $h_u$.

Since $u^2$ acts only on the $u$ factor and the oscillator factors of different summands are orthogonal, every density operator $\rho$ on the shell gives $\operatorname{Tr}(\rho u^2)=w_-a_-+w_0a_0+w_+a_+$, where $w_e\ge0$, with $w_-+w_0+w_+=1$, is the weight of $\rho$ on the summand containing $\psi_e$, and $a_\pm,a_0$ are the values in \eqref{eq:physical-a}. We also prove below that $a_0<a_+<a_-$. Therefore the exact range of $\operatorname{Tr}(\rho u^2)$ over all density operators on the shell is
\begin{equation}\label{eq:physical-u2-range}
a_0\le\operatorname{Tr}(\rho u^2)\le a_-.
\end{equation}
Both endpoints are attained, and convex mixtures realize the whole interval.

\begin{proof}[Completeness of the shell and ordering of the endpoints]
We shall use two simple rational bounds on $\eta$. From \eqref{eq:moment-recurrence},
\begin{align*}
\mu\!\left(\left(u^7-\frac92u^5+6u^3-2u\right)^2\right)&=\frac{4531}{16}\eta-\frac{1081}{8}>0,\\
\mu\!\left(\left(u^8-\frac{11}{2}u^6+8u^4-\frac72u^2\right)^2\right)&=\frac{13791}{32}-\frac{7203}{8}\eta>0.
\end{align*}
Hence
\begin{equation}\label{eq:eta-bounds}
\frac{94}{197}<\eta<\frac{4597}{9604}<\frac12.
\end{equation}

The three even levels $\lambda_0=-4$, $\lambda_2=0$, $\lambda_4=4$ were identified in Section~\ref{Sec:OneD}. The intervening levels $\lambda_1$ and $\lambda_3$ are odd.

For the first odd level, the trial function $u\phi$ gives
\[
\frac{\langle u\phi,h_u u\phi\rangle}{\|u\phi\|_2^2}
=-\frac{3\mu(u^4)}{\mu(u^2)}
=-\frac{3}{2\eta}
<-3,
\]
where the last inequality follows from \eqref{eq:eta-bounds}. Rayleigh--Ritz in the odd subspace therefore yields
\[
-4<\lambda_1<-3.
\]

The second odd level can be localized just as explicitly. Set
\[
P(u)=u-\frac12u^3-\frac{2}{13}u^5,\qquad a=\frac54.
\]
Using \eqref{eq:gauge-rotation} and \eqref{eq:moment-recurrence}, one finds
\[
\frac{1}{16}\mu(P^2)-\mu\!\left(\bigl((\mathscr{D}-a)P\bigr)^2\right)
=\frac{15857-33084\eta}{5408}
>\frac{25435}{6492304}>0,
\]
where \eqref{eq:eta-bounds} was used in the last step. Consequently, by the spectral theorem,
\[
\operatorname{dist}\!\left(a,\operatorname{spec}(h_u|_{\mathrm{odd}})\right)^2
\le
\frac{\|(h_u-a)P\phi\|_2^2}{\|P\phi\|_2^2}
<\frac{1}{16}.
\]
Hence the odd spectrum contains an eigenvalue in $(1,3/2)$. Since $\lambda_1<-3$ and $\lambda_5>\lambda_4=4$, this eigenvalue must be $\lambda_3$. We have therefore established
\begin{equation}\label{eq:odd-level-bounds}
-4<\lambda_1<-3,\qquad 1<\lambda_3<\frac32.
\end{equation}

A level $\lambda$ of $h_u$ contributes to total energy $4$ exactly when $4-\lambda\in\mathbb N_0$. The bounds \eqref{eq:odd-level-bounds} exclude the two odd levels, and all levels above $\lambda_4$ exceed $4$. Only $\lambda=-4,0,4$ contribute, with oscillator energies $8,4,0$, proving \eqref{eq:exact-shell}.

Finally, the differences of the values in \eqref{eq:physical-a} are
\[
a_+-a_0=\frac{12\eta^2+7\eta-6}{4(1-2\eta)},\qquad
a_--a_+=\frac{3-13\eta^2}{1-4\eta^2}.
\]
The first numerator is increasing for $\eta>0$ and is positive at $94/197$. The second is positive at $4597/9604$. Equation~\eqref{eq:eta-bounds} therefore proves $a_0<a_+<a_-$ and the physical interval \eqref{eq:physical-u2-range}.
\end{proof}

\subsection{Persistence of the fixed-observable obstruction}

Let $\sigma_n$ be the normalized orthogonal projector onto $\mathcal H_n$. Tensor the formal functionals of Section~\ref{Sec:OneD} with these oscillator states:
\begin{equation}\label{eq:formal-shell-functional}
\mathcal B_{e,t}=\omega_{e,t}\otimes\sigma_{4-e},\qquad e\in\{-4,0,4\},
\end{equation}
where a density operator in a tensor product of functionals denotes its expectation functional. They are normalized and Hermitian, and satisfy
\begin{equation}\label{eq:formal-shell-identities}
\mathcal B_{e,t}(A(H-4))
=\mathcal B_{e,t}((H-4)A)
=\mathcal B_{e,t}([H,A])=0
\quad(A\in\mathcal A).
\end{equation}
For each finite $K$, their moment matrices are positive semidefinite when $t$ is sufficiently close to $\eta$. To see this, apply $A$ to each polynomial-times-Gaussian basis vector of $\psi_e\otimes\mathcal H_{4-e}$. After conjugating by
\[
\Gamma=\exp[-u^4/4-(x^2+y^2+z^2)/2],
\]
expand the resulting polynomial in oscillator Hermite polynomials. The coefficients are polynomials in $u$ of degree at most $3K+4$ when $\deg A\le K$. Integrating the oscillator variables expresses $\mathcal B_{e,t}(A^\dagger A)$ as a positive combination of terms $\mu_t(|f(u)|^2)/\mu_t(P_e^2)$. Thus $J_{3K+4}(t)\succ0$ suffices.

Let
\begin{equation}\label{eq:shell-annihilator}
\mathcal I_4=\{A\in\mathcal A:A\psi=0\text{ for every }\psi\in\ker(H-4)\}.
\end{equation}

\begin{theorem}[Finite constraints on the four-mode shell]\label{thm:finite-strengthening}
Let $\mathcal W\subset\mathcal A$ be a finite-dimensional $*$-closed space containing $1,u^2$, and let $G_1,\ldots,G_m$, $m\ge1$, be polynomial operators with $G_i^\dagger G_j\in\mathcal W$. There are normalized Hermitian functionals $L^\pm$ on $\mathcal W$ satisfying
\[
\bigl(L^\pm(G_i^\dagger G_j)\bigr)_{i,j}\succeq0,\qquad
L^\pm\bigl(\mathcal W\cap(\mathcal I_4+\mathcal I_4^\dagger)\bigr)=0,
\]
and $L^-(u^2)<a_0$, $L^+(u^2)>a_-$. In particular, for every $K\ge3$,
\begin{equation}\label{eq:strict-u2-bounds}
\inf_{L\in\mathcal F_K(4)}L(u^2)<a_0,\qquad
\sup_{L\in\mathcal F_K(4)}L(u^2)>a_-.
\end{equation}
Writing $\Sigma_{\mathcal A}=\{\sum_j A_j^\dagger A_j\}$ for finite sums, neither endpoint has a finite polynomial-square certificate:
\begin{equation}\label{eq:no-finite-certificate}
u^2-a_0\notin\Sigma_{\mathcal A}+\mathcal I_4+\mathcal I_4^\dagger,\qquad
a_--u^2\notin\Sigma_{\mathcal A}+\mathcal I_4+\mathcal I_4^\dagger.
\end{equation}
\end{theorem}

\begin{proof}
Set $D=\max_i\deg G_i$. Since $\sum_{i,j}\bar c_ic_j\mathcal B_{e,t}(G_i^\dagger G_j)=\mathcal B_{e,t}(A^\dagger A)$ with $A=\sum_ic_iG_i$ of degree at most $D$, the preceding argument shows that the Gram matrix is positive semidefinite whenever $J_{3D+4}(t)\succ0$, hence for $t$ sufficiently close to $\eta$. If $A\in\mathcal I_4$, then $\widehat A=\Gamma^{-1}A\Gamma$ annihilates each seed polynomial $P_e(u)H_{\mathbf n}(x,y,z)$ with $|\mathbf n|=4-e$. This is a polynomial identity independent of $t$, so $\mathcal B_{e,t}$ vanishes on $\mathcal I_4$ and, by Hermiticity, on $\mathcal I_4^\dagger$. Choose $L^-=\mathcal B_{0,t^-}|_{\mathcal W}$ and $L^+=\mathcal B_{-4,t^+}|_{\mathcal W}$ with $t^-<\eta<t^+$ sufficiently close to $\eta$. Their $u^2$ values are $a_0(t^-)$ and $a_-(t^+)$, which give the strict inequalities. Equation~\eqref{eq:formal-shell-identities} supplies the standard hierarchy constraints. The certificate statement follows by applying these functionals to any proposed finite identity, as in Corollary~\ref{cor:oneD-no-certificate}.
\end{proof}

All stationarity polynomials already lie in $\mathcal I_4+\mathcal I_4^\dagger$, by the identity \eqref{eq:stationarity-ideal} with $h_u-e$ replaced by $H-4$. The theorem also covers finite localizing tests with Hermitian weights $g$ satisfying $\langle f\Gamma,gf\Gamma\rangle>0$ for every nonzero polynomial $f$ in four variables: apply the finite-image continuity argument of Proposition~\ref{prop:localizing} to the shell seed polynomials. In particular, every nonzero nonnegative position polynomial has this property.

The optimization values in \eqref{eq:strict-u2-bounds} are finite. Since $N$ is a sum of oscillator number operators, $L(N)\ge0$. The identity used in \eqref{eq:oneD-bounded} and $L(H)=4$ then give
\begin{equation}\label{eq:u2-uniform-finite}
0\le L(u^2)\le24\qquad(L\in\mathcal F_K(4),\ K\ge3).
\end{equation}

\subsection{Position moments without a representing measure}

We now construct feasible moments whose position restriction cannot be represented by a positive Borel measure on $\mathbb R^4$. Unlike the fixed-observable deformation, this construction changes only moments at the highest retained degree.

We begin with the Robinson sextic
\begin{equation}\label{eq:Robinson}
R(X,Y,Z)=X^6+Y^6+Z^6-X^4Y^2-X^4Z^2-Y^4X^2-Y^4Z^2-Z^4X^2-Z^4Y^2+3X^2Y^2Z^2,
\end{equation}
a standard example of a nonnegative polynomial that is not a sum of polynomial squares \cite{Robinson1973,Reznick2000}. Its nonnegativity follows directly by setting $a=X^2$, $b=Y^2$, $c=Z^2$ and, by symmetry, assuming $a\ge b\ge c\ge0$:
\begin{equation}\label{eq:Robinson-positive}
R=(a-b)^2(a+b-c)+c(a-c)(b-c)\ge0.
\end{equation}
On the homogeneous sextics define a real linear functional $\lambda$ by permutation symmetry, by setting every monomial with an odd exponent to zero, and by
\begin{equation}\label{eq:lambda-values}
\lambda(X^6)=\frac{37}{20},\qquad \lambda(X^4Y^2)=1,\qquad \lambda(X^2Y^2Z^2)=\frac{1}{10}.
\end{equation}
On homogeneous cubics the quadratic form $Q\mapsto\lambda(Q^2)$ decomposes according to exponent parity. For example, on the ordered basis $\{X^3,XY^2,XZ^2\}$ the block is
\[
\begin{pmatrix}
37/20&1&1\\
1&1&1/10\\
1&1/10&1
\end{pmatrix}.
\]
The bases $\{Y^3,YX^2,YZ^2\}$ and $\{Z^3,ZX^2,ZY^2\}$ give the same block by permutation symmetry, while $XYZ$ gives the scalar block $1/10$. The successive principal minors of the $3\times3$ block are $37/20$, $17/20$ and $63/2000$, so the form is positive definite. On the other hand,
\begin{equation}\label{eq:lambda-R}
\lambda(R)=3\frac{37}{20}-6+\frac{3}{10}=-\frac{3}{20}.
\end{equation}

Define a functional $\ell_3$ on polynomials in $X,Y$ of degree at most six by
\begin{equation}\label{eq:ell3-definition}
\ell_3(X^aY^b)=\lambda(X^aY^bZ^{6-a-b}),\qquad a+b\le6.
\end{equation}
If $q(X,Y)$ has degree at most three and $Q(X,Y,Z)=Z^3q(X/Z,Y/Z)$ is its homogenization, then $\ell_3(q^2)=\lambda(Q^2)$. For a functional $\ell$ on $\mathbb R[X,Y]$, let $M_K(\ell)$ denote the commutative moment matrix indexed by monomials $X^{\alpha_1}Y^{\alpha_2}$ of degree at most $K$, with entries
\[
M_K(\ell)_{\alpha\beta}=\ell\!\left(X^{\alpha_1+\beta_1}Y^{\alpha_2+\beta_2}\right).
\]
Hence
\begin{equation}\label{eq:ell3-properties}
M_3(\ell_3)\succ0,\qquad \ell_3(R(X,Y,1))=-\frac{3}{20}.
\end{equation}
Thus $\ell_3$ is strictly positive on squares yet negative on the nonnegative polynomial $R(X,Y,1)$, i.e.\ it separates $R$ from the sums of squares. This functional can be extended to arbitrary order without losing either property.

 Suppose $\ell_k$ is defined through degree $2k$ and $M_k(\ell_k)\succ0$. After ordering the monomials of degree at most $k$ before those of degree $k+1$, write
\[
M_{k+1}=\begin{pmatrix}A&B\\B^{\mathsf T}&C\end{pmatrix},\qquad A=M_k(\ell_k)\succ0.
\]
Keep all old moments fixed, set the new moments of degree $2k+1$ to zero, and choose the moments of degree $2k+2$ to be a sufficiently large positive multiple of the corresponding moments of the centered standard Gaussian probability measure on $\mathbb{R}^2$. Then $C=tG$ with $G\succ0$, and for sufficiently large $t$ the Schur complement $tG-B^{\mathsf T}A^{-1}B$ is positive definite. Iteration therefore gives, for every $K\ge3$, an extension $\ell_K$ of $\ell_3$ satisfying
\begin{equation}\label{eq:ellK-positive}
M_K(\ell_K)\succ0.
\end{equation}
No flatness condition is imposed here: positivity of a truncated moment matrix alone does not imply the existence of a representing measure, in contrast with flat-extension criteria \cite{CurtoFialkow1998}.

We now lift this finite-dimensional separation to the Weyl moment problem. For a polynomial operator $A$, write it in fully symmetrized Weyl order and identify the resulting coefficients with its commuting Weyl symbol. Let $A_{2K}^{\mathrm{osc}}(q_{\mathrm{osc}},p_{\mathrm{osc}})$ be the homogeneous degree-$2K$ part of that symbol that depends only on
\[
q_{\mathrm{osc}}:=(x,y,z),\qquad p_{\mathrm{osc}}:=(p_x,p_y,p_z),
\]
with no dependence on $u$ or $p_u$. Reordering between PBW and Weyl form changes only lower-degree terms: every use of $[q_j,p_j]=\mathrm{i}$ replaces a degree-two pair by a scalar and therefore lowers total degree by two. Equivalently, the associated graded Weyl algebra is the ordinary commutative polynomial algebra. Thus the top homogeneous symbol is independent of the ordering convention and is multiplicative. In particular, for $\deg A\le K$, if $a_K^{\mathrm{osc}}$ denotes the degree-$K$ purely oscillator part of the top symbol of $A$, then
\[
(A^\dagger A)_{2K}^{\mathrm{osc}}=\overline{a_K^{\mathrm{osc}}}\,a_K^{\mathrm{osc}}.
\]
Set
\[
a_K=\int_0^{2\pi}\frac{\cos^{2K}\theta}{2\pi}\,d\theta=4^{-K}\binom{2K}{K},
\]
and define
\begin{equation}\label{eq:tail-functional}
T_K(A)=\frac{1}{a_K}\int_0^{2\pi}\frac{d\theta}{2\pi}\,
\ell_K\!\left(A_{2K}^{\mathrm{osc}}(\xi\cos\theta,-\xi\sin\theta)\big|_{\xi_z=1}\right),
\end{equation}
where $\xi=(\xi_x,\xi_y,\xi_z)$ and $\ell_K$ acts on the remaining variables $\xi_x,\xi_y$. The functional $T_K$ vanishes on every operator of degree below $2K$. Since $\ell_K$ is real and adjunction complex-conjugates the top Weyl symbol, $T_K(A^\dagger)=\overline{T_K(A)}$. Moreover,
\begin{equation}\label{eq:tail-positivity}
T_K(A^\dagger A)\ge0,\qquad \deg A\le K.
\end{equation}
Indeed, the displayed top-symbol identity shows that only $a_K^{\mathrm{osc}}$ contributes. After the phase rotation and dehomogenization in \eqref{eq:tail-functional}, the integrand is $\ell_K(|b_\theta|^2)$ for a polynomial $b_\theta$ of degree at most $K$. Writing $b_\theta=b_{\theta,1}+\mathrm{i}b_{\theta,2}$ reduces this to $\ell_K(b_{\theta,1}^2)+\ell_K(b_{\theta,2}^2)\ge0$ by \eqref{eq:ellK-positive}.

The phase average also gives the required stationarity. Let $\widetilde T_K$ be the functional on the full oscillator Weyl algebra obtained from the same formula by retaining the homogeneous Weyl-symbol component of degree $2K$. 
Conjugation by $e^{\mathrm{i}sN}$ rotates every oscillator pair $(q_j,p_j)$ in phase space by the same angle $s$; since $N$ is quadratic, it acts on Weyl symbols exactly by this rotation of their arguments, and it preserves degree. The $\theta$-average is invariant under such rotations, so $\widetilde T_K(e^{\mathrm{i}sN}Be^{-\mathrm{i}sN})=\widetilde T_K(B)$ for all $s$, and differentiating at $s=0$ gives
\[
\widetilde T_K([N,B])=0
\]
for every oscillator polynomial $B$. Choose any normalized eigenstate functional $\omega_e$ of $h_u$. The tensor-product functional $\omega_e\otimes\widetilde T_K$ is stationary for $h_u+N$: the $h_u$ commutator vanishes under $\omega_e$, and the $N$ commutator vanishes under $\widetilde T_K$. Its restriction to $\mathcal A_{4,\le2K}$ is exactly $T_K$. Indeed, a monomial containing $u$ or $p_u$ has oscillator degree strictly below $2K$ and is annihilated by $\widetilde T_K$; a purely oscillator monomial of degree below $2K$ is annihilated for the same reason; and on purely oscillator degree-$2K$ terms the two definitions coincide, with $\omega_e(1)=1$. Consequently,
\begin{equation}\label{eq:tail-stationarity}
T_K([H,O])=0
\end{equation}
for every $O\in\mathcal A$ whose PBW-reduced commutator lies in $\mathcal A_{4,\le2K}$. Finally, if $\deg O\le2K-6$, then
\begin{equation}\label{eq:tail-eigenstate}
T_K(O(H-4))=0.
\end{equation}
All contributions involving terms of degree at most two in $H$ lie below degree $2K$, while the only possible degree-$2K$ contribution from $u^6$ still contains a $u$ variable and is removed by the purely oscillator projection.

The density operator
\begin{equation}\label{eq:background-state}
\rho_*=\frac13\left(|\psi_{-4}\rangle\langle\psi_{-4}|\otimes\sigma_8+|\psi_0\rangle\langle\psi_0|\otimes\sigma_4+|\psi_4\rangle\langle\psi_4|\otimes\sigma_0\right)
\end{equation}
is supported on $\ker(H-4)$. Denote its moment functional by $B_\eta$.

\begin{theorem}[Nonphysical feasible moments at every finite order]\label{thm:allK}
For every $K\ge3$, there exists $L_K\in\mathcal F_K(4)$ whose position moments admit no positive representing measure. In particular, $L_K$ cannot be the truncation of a Schr\"odinger quantum state.
\end{theorem}

\begin{proof}
Set
\begin{equation}\label{eq:finite-witness}
f_K=z^{2K-6}R(x,y,z),\qquad c_K=B_\eta(f_K),\qquad \tau_K=\frac{20}{3}c_K+1,
\end{equation}
and define $L_K=B_\eta+\tau_KT_K$. Since $f_K\ge0$, we have $c_K\ge0$ and hence $\tau_K>0$. Equations \eqref{eq:tail-positivity}--\eqref{eq:tail-eigenstate} show that $L_K$ satisfies positivity, stationarity and all eigenstate relations defining $\mathcal F_K(4)$, while $T_K(1)=0$ preserves normalization. Since $f_K$ is homogeneous of degree $2K$, the phase average contributes the factor $a_K$, and because $\ell_K$ extends $\ell_3$,
\[
T_K(f_K)=\ell_K(R(X,Y,1))=-\frac{3}{20}.
\]
Therefore
\[
L_K(f_K)=c_K-\frac{3}{20}\left(\frac{20}{3}c_K+1\right)=-\frac{3}{20}.
\]
A positive measure on the commuting position variables cannot assign a negative expectation to the pointwise nonnegative polynomial $f_K$. Hence the position moments of $L_K$ have no positive representing measure.
\end{proof}

The construction is order dependent: $f_K$, $\ell_K$ and $T_K$ change with $K$, while every moment of degree below $2K$ agrees with the physical background $B_\eta$. In particular, $L_K$ converges moment by moment to that physical functional, consistently with Theorem~\ref{thm:reconstruction}.

\begin{proposition}[Interior nonphysical points at level three]
\label{prop:interior-tail}
The functional $L_3$ constructed in Theorem~\ref{thm:allK} satisfies
$M_3(L_3)\succ0$. Moreover, within the real affine space defined by
normalization, Hermiticity and the linear identities of
$\mathcal F_3(4)$, there is a nonempty relatively open set of feasible
functionals whose position moments admit no positive representing
measure.
\end{proposition}

\begin{proof}
We first show that $M_3(B_\eta)\succ0$. Let $\deg A\le3$ and suppose
that $B_\eta(A^\dagger A)=0$. Since $\rho_*$ is strictly positive on
the entire finite-dimensional shell $\ker(H-4)$, the operator $A$
annihilates that shell. Write
\[
A=\sum_{a+b\le3}u^ap_u^b\otimes B_{ab},
\qquad \deg B_{ab}\le3-a-b,
\]
where the $B_{ab}$ involve only the three oscillator modes.

For each $e\in\{-4,0\}$, the ten polynomials
$\phi^{-1}u^ap_u^b(P_e\phi)$, with $a+b\le3$, are linearly independent.
The four even members have distinct degrees $4,6,8,10$.
Among the six odd members, replacing the operator $p_u$ by
$p_u-\mathrm{i}u^3$ gives distinct degrees $3,5,7,9,11,13$,
since
\[
(p_u-\mathrm{i}u^3)(P_e\phi)=-\mathrm{i}P_e'\phi.
\]
Applying $A$ to $\psi_{-4}\otimes\mathcal H_8$ and
$\psi_0\otimes\mathcal H_4$ therefore shows that every $B_{ab}$
annihilates both $\mathcal H_8$ and $\mathcal H_4$.

No nonzero oscillator polynomial $B$ of degree at most three annihilates both
$\mathcal H_8$ and $\mathcal H_4$. To see this, use the Bargmann--Fock
representation, in which the creation and annihilation operators act as
multiplication by $w_j$ and as $\partial_{w_j}$, respectively, and
$\mathcal H_n$ consists of the homogeneous polynomials of degree $n$ in $w$.
Write $B=\sum_k B_k$, where $B_k$ collects the normally ordered terms that
raise the oscillator number by $k$. Since $B_k$ maps $\mathcal H_n$ into
$\mathcal H_{n+k}$, $B$ annihilates $\mathcal H_n$ only if every $B_k$ does.
Because $\deg B\le3$, the components with $k=\pm2,\pm3$ are pure creation or
pure annihilation polynomials: the former act by multiplication by a nonzero
polynomial and the latter by a constant-coefficient differential operator of
order two or three, so neither annihilates $\mathcal H_4$ unless it is zero.
The components with $k=0$ and $k=+1$ have the respective forms
\[
F_1(w)\cdot\partial+c,
\qquad
F_2(w)\cdot\partial+g_1(w),
\]
where $F_1,F_2$ are homogeneous vector polynomials of degrees one
and two, and $g_1$ is linear. Testing on $w_j^n$, for
$j=1,2,3$, gives respectively
\[
nF_1+cw=0,
\qquad
nF_2+g_1w=0.
\]
Using both $n=8$ and $n=4$ forces all these coefficients to vanish.
The component with $k=-1$ vanishes as well: since
$\langle B_{-1}\phi,\psi\rangle=\langle\phi,B_{-1}^\dagger\psi\rangle$ for
$\phi\in\mathcal H_n$ and $\psi\in\mathcal H_{n-1}$, it annihilates
$\mathcal H_8$ and $\mathcal H_4$ exactly when its adjoint, a $k=+1$ polynomial
of degree at most three, annihilates $\mathcal H_7$ and $\mathcal H_3$, and the
$k=+1$ argument with $n=7,3$ then gives $B_{-1}^\dagger=0$.
Hence every $B_{ab}$ is zero, and consequently $A=0$.

This proves $M_3(B_\eta)\succ0$. Since $\tau_3>0$ and
$M_3(T_3)\succeq0$,
\[
M_3(L_3)
=M_3(B_\eta)+\tau_3M_3(T_3)
\succ0.
\]
The construction also gives $L_3(R)=-3/20$.
Positive definiteness and this strict negative inequality persist
in a sufficiently small relative neighborhood satisfying the
linear constraints. Every functional in that neighborhood is
feasible and assigns a negative value to the nonnegative position
polynomial $R$, so its position moments admit no positive
representing measure.
\end{proof}

\subsection{Excluding the tail by finite constraints}

One localizing inequality excludes the functional in Theorem~\ref{thm:allK}. Since $R(x,y,z)\ge0$, every physical state satisfies
\begin{equation}\label{eq:tail-localizer}
L\bigl((z^{K-3})^\dagger R(x,y,z)z^{K-3}\bigr)\ge0.
\end{equation}
Its left-hand side is $L(f_K)$, whereas the constructed $L_K$ gives $-3/20$. A single entry of the localizing matrix for the fixed weight $R$ therefore removes this family at each level. The strengthened hierarchy still has nonphysical bounds on $u^2$ after any finite collection of these tests, by Theorem~\ref{thm:finite-strengthening}.

An exact shell relation removes the same tail. Set
\begin{equation}\label{eq:shell-polynomial}
F(N)=N(N-4)(N-8).
\end{equation}
It annihilates \eqref{eq:exact-shell}. Let $\mathcal F_K^{\mathrm{sh}}(4)$ impose, in addition to the standard constraints,
\begin{equation}\label{eq:shell-polynomial-constraints}
L(OF(N))=L(F(N)O)=0,\qquad\deg O\le2K-6.
\end{equation}
Taking $O=N^{K-3}$ gives
\begin{equation}\label{eq:tail-cut}
T_K(N^{K-3}F(N))=T_K(N^K)
=\frac{1}{2^Ka_K}\ell_K((1+X^2+Y^2)^K)>0.
\end{equation}
The lower powers of $N$ have degree below $2K$, while $(1+X^2+Y^2)^K$ is a nonzero sum of squares tested by $M_K(\ell_K)\succ0$. The formal functionals $\mathcal B_{e,t}$ satisfy all these shell relations, because their oscillator factors are supported at $N=4-e$. Hence the fixed-observable bounds remain strict in $\mathcal F_K^{\mathrm{sh}}(4)$.

\subsection{Convergence on the same shell}

For the potential \eqref{eq:main-potential}, set
\[
s=x^2+y^2+z^2,\qquad U=V_*+13.
\]
Then
\begin{equation}\label{eq:main-confinement-certificate}
U-\frac14(u^6+s)
=
\frac14(u^3-2u)^2
+\left(u^2-\frac{13}{4}\right)^2
+\frac{15}{16}
+\frac14s.
\end{equation}
Thus \eqref{eq:sos-potential} and \eqref{eq:sos-confinement} hold with
\begin{gather*}
v_p=v_q=-13,\qquad r=3,\qquad r_u=3,\qquad r_x=r_y=r_z=1,\\
\epsilon_j=\frac14\qquad(j=u,x,y,z).
\end{gather*}
The identity \eqref{eq:oneD-derivative}, with $U_u$ replaced by $U$, verifies the second-derivative bound in $u$, since \eqref{eq:main-confinement-certificate} gives $U-\frac14u^6\in\Sigma_6$. The higher $u$ derivatives are unchanged, so we may use the same coefficients:
\[
(a_{u1},b_{u1})=(40,0),\qquad
(a_{u2},b_{u2})=(0,0),\qquad
(a_{u3},b_{u3})=(0,360).
\]
For $j=x,y,z$, $\partial_j^2V_*=1$ and all higher even derivatives vanish; take $(a_{j1},b_{j1})=(0,1)$ and zero coefficients at higher orders. Together with essential self-adjointness, these certificates establish the hypotheses of Theorem~\ref{thm:reconstruction} at $E=4$.

Define the finite-order bounds
\begin{equation}\label{eq:u2-bootstrap-endpoints}
\underline b_K=\inf_{L\in\mathcal F_K(4)}L(u^2),\qquad
\overline b_K=\sup_{L\in\mathcal F_K(4)}L(u^2).
\end{equation}
Theorem~\ref{thm:finite-strengthening} gives, for every finite $K\ge3$,
\begin{equation}\label{eq:u2-finite-strict}
\underline b_K
<
a_0=\frac32\eta,
\qquad
\overline b_K
>
a_-=\frac{13\eta+6}{4+8\eta}.
\end{equation}
Restriction from level $K+1$ to level $K$ preserves feasibility, so $\underline b_K$ is nondecreasing and $\overline b_K$ is nonincreasing. Theorem~\ref{thm:reconstruction}, applied with $D=2$, now gives
\begin{equation}\label{eq:u2-limit}
\underline b_K\uparrow a_0,\qquad
\overline b_K\downarrow a_-.
\end{equation}
The same conclusion holds for the shell-strengthened hierarchy $\mathcal F_K^{\mathrm{sh}}(4)$.

\section{Discussion}\label{Sec:Discussion}

For the sextic oscillator, the bootstrap bounds on the fixed observable $u^2$ converge to its physical value, yet differ from it at every finite order. Its physical value has no finite polynomial-square certificate even modulo the full eigenstate annihilator. Finite exactness therefore requires more than explicit knowledge of the energy and eigenfunction. Additional algebraic structure can make a difference, as in harmonic-oscillator and shape-invariant bootstraps \cite{Aikawa2022,Aikawa2025}, while optimality conditions provide further constraints in noncommutative polynomial optimization \cite{Araujo2026}.

Both the four-mode construction at level three
(Proposition~\ref{prop:interior-tail}) and the companion example
\cite{KerenZhou2026Interior} give strictly positive definite spurious
moment matrices. The companion paper uses a pure eigenstate
background for a coupled three-dimensional double well at its
first excited energy, and also at its ground-state energy, with
nonphysicality certified by a nonnegative momentum sextic.
Here the background is a mixed state on an explicitly known
four-mode excited shell, and the witness is a nonnegative position
polynomial. The constraint conventions also differ. The companion
paper imposes stationarity rows generated by individual monomials
and verifies additional cancellation relations from multipliers
of degree at most eight. In \eqref{eq:finite-hierarchy}, the
stationarity condition includes arbitrary polynomial generators
whose reduced commutators lie in the retained moment space,
without a bound on their degree.

The four-mode example distinguishes this persistent obstruction from a removable one. Its Robinson tail produces position moments with no positive representing measure, but a single localizing inequality excludes the constructed point. The fixed-observable deformation survives such finite tests. Both constructions depend on the truncation order, consistently with convergence to the physical moment set.

The convergence theorem allows coupled polynomial potentials, although the derivative-domination certificates may be restrictive. Whether weaker hypotheses still yield uniform moment bounds remains open. The compactness proof also gives no convergence rate. In simplified bootstrap problems, rapid convergence is related to factorial moment growth and small moment-matrix eigenvalues \cite{Berenstein2025Goldilocks}. Relating these estimates to the finite deformations above may give quantitative bounds for fixed observables.

\appendix

\section{Numerical bounds}
\label{app:numerics}
Table~\ref{tab:bounds} reports numerical estimates of the infimum and supremum of
$L(u^2)$ over $\mathcal F_K^h(e)$ in \eqref{eq:finite-hierarchy}.
Eigenstate multipliers have degree at most $2K-6$; stationarity includes every
polynomial generator whose reduced commutator has degree at most $2K$.
Parity and time reversal preserve feasibility and $L(u^2)$, so we may restrict
to functionals vanishing on monomials of odd total degree and satisfying
$L(u^ip_u^j)=(-\mathrm i)^j r_{ij}$, $r_{ij}\in\mathbb R$.
In the basis $\mathrm i^ju^ip_u^j$, the moment matrix is real and splits into
even and odd blocks. All linear equalities are solved exactly over $\mathbb Q$,
leaving seven free variables at $K=3$ and eight at $K=4,5$.
The implementation's commutator cutoff was checked against the complete
weighted-degree bound $6K-2$, with weights one for $u$ and three for $p_u$;
both give the same reduced affine space in all nine cases.
The bound and the exact comparison are documented in the ancillary files.

These programs are poorly conditioned for double-precision solvers, whose
reported optima varied by several per cent with the scaling even at $K=3$.
We therefore used SDPA-GMP~\cite{Nakata2010} in 256-bit arithmetic.
The physical moments, evaluated at 60 decimal digits, satisfy the reduced
affine constraints with maximum residual below $2\times10^{-52}$.
For $K=3,4$, the absolute primal--dual objective differences are at most
$2.1\times10^{-4}$; all six runs at $K=5$ return \texttt{pdOPT}, with matching
objectives at the printed precision. These solver diagnostics are not rigorous
error bounds: the table gives numerical estimates, while the strict inequalities
and convergence in \eqref{eq:oneD-limits} are proved analytically.
The scripts, SDP inputs, solver outputs and independent checks are supplied
as ancillary files.

From $K=6$ on, $h_u-e$ lies in the kernel of every feasible moment matrix,
so the problem has no strictly feasible point. Removing this kernel exactly
(facial reduction) was not sufficient for the solver at $K=6$; we do not
report those levels.

\bibliography{reference}

\end{document}